\documentclass[11pt,onecolumn]{article}

\usepackage{amsmath,amsthm,amssymb,amsfonts,epsfig,graphicx}
\usepackage[linesnumbered,ruled,vlined]{algorithm2e}
\usepackage{color}
\usepackage{dsfont}
\usepackage{enumerate}
\usepackage{url}
\usepackage{filecontents}

\usepackage[left=0.9in,top=0.9in,right=0.9in,bottom=0.5in,nohead,textheight=10in,footskip=0.3in]{geometry}

\usepackage{lipsum}

\let\OLDthebibliography\thebibliography
\renewcommand\thebibliography[1]{
  \OLDthebibliography{#1}
  \setlength{\itemsep}{0pt}
}

\newtheorem{theorem}{Theorem}

\newtheorem{lemma}[theorem]{Lemma}

\newtheorem{example}{Example}
\newtheorem{remark}{Remark}

\begin{document}

\def\myparagraph#1{\vspace{2pt}\noindent{\bf #1~~}}
\newcommand{\eqdef}{{\stackrel{\mbox{\tiny \tt ~def~}}{=}}}



\def\DeltaCeil{{\lceil\Delta\rceil}}
\def\TwoDeltaCeil{{\lceil 2\Delta\rceil}}
\def\OnePointFiveDeltaCeil{{\lceil 3\Delta/2\rceil}}

\long\def\ignore#1{}
\def\myps[#1]#2{\includegraphics[#1]{#2}}
\def\etal{{\em et al.}}
\def\Bar#1{{\bar #1}}
\def\br(#1,#2){{\langle #1,#2 \rangle}}
\def\setZ[#1,#2]{{[ #1 .. #2 ]}}
\def\Pr{\mbox{\rm Pr}}
\def\REACHED{\mbox{\tt REACHED}}
\def\AdjustFlow{\mbox{\tt AdjustFlow}}
\def\GetNeighbors{\mbox{\tt GetNeighbors}}
\def\true{\mbox{\tt true}}
\def\false{\mbox{\tt false}}
\def\Process{\mbox{\tt Process}}
\def\ProcessLeft{\mbox{\tt ProcessLeft}}
\def\ProcessRight{\mbox{\tt ProcessRight}}
\def\Add{\mbox{\tt Add}}

\def\maxsuperscript[#1]{\raisebox{#1}{\scriptsize $({\max})$}}

\def\bmin{{\beta_{\min}}}
\def\bmax{{\beta_{\max}}}

\def\bP{{\bf P}}

\def\E{{\mathbb E}}
\def\P{{\mathbb P}}
\def\X{{\mathbb X}}
\def\Y{{\mathbb Y}}
\def\Z{{\mathbb Z}}
\def\W{{\mathbb W}}
\def\U{{\mathbb U}}
\def\V{{\mathbb V}}
\def\S{{\mathbb S}}
\def\R{{\mathbb R}}

\def\Vrel(#1){{{\mathbb S}[{#1}]}}

\def\Xs{{\tiny \mbox{$\mathbb X$}}}
\def\Ys{{\tiny \mbox{$\mathbb Y$}}}
\def\Zs{{\tiny \mbox{$\mathbb Z$}}}
\def\Ws{{\tiny \mbox{$\mathbb W$}}}
\def\Us{{\tiny \mbox{$\mathbb U$}}}
\def\Vs{{\tiny \mbox{$\mathbb V$}}}
\def\Uts{{\tiny \mbox{$\mathbb U_t$}}}
\def\Ualphas{{\tiny \mbox{$\mathbb U_\alpha$}}}
\def\Ss{{\tiny \mbox{$\mathbb S$}}}
\def\Pis{{\tiny \mbox{$\Pi$}}}

\def\setof#1{{\left\{#1\right\}}}
\def\suchthat#1#2{\setof{\,#1\mid#2\,}} 
\def\event#1{\setof{#1}}
\def\q={\quad=\quad}
\def\qq={\qquad=\qquad}
\def\calA{{\cal A}}
\def\calB{{\bf B}}
\def\calC{{\cal C}}
\def\calD{{\cal D}}
\def\calE{{\cal E}}
\def\calF{{\cal F}}
\def\calG{{\cal G}}
\def\calI{{\cal I}}
\def\calH{{\cal H}}
\def\calL{{\cal L}}
\def\calN{{\cal N}}
\def\calP{{\cal P}}
\def\calQ{{\cal Q}}
\def\calR{{\cal R}}
\def\calS{{\cal S}}
\def\calT{{\cal T}}
\def\calU{{\cal U}}
\def\calV{{\cal V}}
\def\calO{{\cal O}}
\def\calX{{\cal X}}
\def\calY{{\cal Y}}
\def\calZ{{\cal Z}}
\def\psfile[#1]#2{}
\def\psfilehere[#1]#2{}
\def\epsfw#1#2{\includegraphics[width=#1\hsize]{#2}}
\def\assign(#1,#2){\langle#1,#2\rangle}
\def\edge(#1,#2){(#1,#2)}
\def\VS{\calV^s}
\def\VT{\calV^t}
\def\slack(#1){\texttt{slack}({#1})}
\def\barslack(#1){\overline{\texttt{slack}}({#1})}
\def\NULL{\texttt{NULL}}
\def\PARENT{\texttt{PARENT}}
\def\GRANDPARENT{\texttt{GRANDPARENT}}
\def\TAIL{\texttt{TAIL}}
\def\HEADORIG{\texttt{HEAD$\_\:$ORIG}}
\def\TAILORIG{\texttt{TAIL$\_\:$ORIG}}
\def\HEAD{\texttt{HEAD}}
\def\CURRENTEDGE{\texttt{CURRENT$\!\_\:$EDGE}}

\def\unitvec(#1){{{\bf u}_{#1}}}
\def\uvec{{\bf u}}
\def\vvec{{\bf v}}
\def\Nvec{{\bf N}}

\newcommand{\bg}{\mbox{$\bf g$}}
\newcommand{\bh}{\mbox{$\bf h$}}

\newcommand{\bx}{\mbox{\boldmath $x$}}
\newcommand{\by}{\mbox{\boldmath $y$}}
\newcommand{\bz}{\mbox{\boldmath $z$}}
\newcommand{\bu}{\mbox{\boldmath $u$}}
\newcommand{\bv}{\mbox{\boldmath $v$}}
\newcommand{\bw}{\mbox{\boldmath $w$}}
\newcommand{\bvarphi}{\mbox{\boldmath $\varphi$}}

\newcommand\myqed{{}}


\title{\Large\bf  \vspace{-30pt} A Faster Approximation Algorithm for the Gibbs Partition Function}
\author{Vladimir Kolmogorov \\ \normalsize Institute of Science and Technology Austria \\ {\normalsize\tt vnk@ist.ac.at}}
\date{}
\maketitle

\begin{abstract}
We consider the problem of estimating the partition function $Z(\beta)=\sum_x \exp(-\beta H(x))$
of a Gibbs distribution with a Hamilton $H(\cdot)$, or more precisely the logarithm of the ratio $q=\ln Z(0)/Z(\beta)$.
It has been recently shown how to approximate $q$ with high probability assuming
the existence of an oracle that produces samples from the Gibbs distribution for a given parameter value in $[0,\beta]$.
The current best known approach due to Huber~\cite{Huber:Gibbs} uses $O(q\ln n\cdot[\ln q + \ln \ln n+\varepsilon^{-2}])$ 
oracle calls on average where $\varepsilon$ is the desired accuracy of approximation and $H(\cdot)$ is assumed to lie in $\{0\}\cup[1,n]$.
We improve the complexity to $O(q\ln n\cdot\varepsilon^{-2})$ oracle calls.
We also show that the same complexity can be achieved if exact oracles are replaced with approximate sampling oracles
that are within $O(\frac{\varepsilon^2}{q\ln n})$
variation distance from exact oracles.
Finally, we prove a lower bound of $\Omega(q\cdot \varepsilon^{-2})$ oracle calls under a natural model of computation.
\end{abstract}

\section{Introduction}
It is known that for large classes of problems, e.g.\ {\em self-reducible problems}~\cite{Jerrum:86},
there is an intimate connection between approximate counting and sampling: the ability to solve
one problem allows solving the other one. This paper explores this connection for Gibbs distributions.

Let $\Omega$ be some finite set and $H(\cdot)$ be some real-valued function on $\Omega$
called a {\em Hamiltonian}. The {\em Gibbs distribution} for such a system is a family of distributions $\{\mu_\beta\}$ on $\Omega$ parameterized by $\beta$, where 
\begin{equation}
\mu_\beta(x)=\frac{1}{Z(\beta)} \exp(-\beta H(x))\qquad \forall x\in\Omega   \label{eq:Gibbs}
\end{equation}
The normalizing constant $Z(\beta)$ is called the {\em partition function}:
\begin{equation}
Z(\beta)=\sum_{x\in\Omega} \exp(-\beta H(x))
\end{equation}
Estimating this function for a given value of $\beta$ is a widely studied computational problem
with applications in many areas. 
In particular, it is a key computational task in statistical physics.
Evaluations of $Z(\cdot)$ yield estimates of important thermodynamical quantities, such as the free energy. 
Note, parameter $\beta$ corresponds to the {\em inverse temperature}.
A classical example of a Gibbs distribution in physics
is  the {\em Ising model}.
\begin{example}\label{ex:Ising}
Given an undirected graph $(V,E)$, let $\Omega=\{-1,+1\}^V$ and $H(x)=\sum_{\{i,j\}\in E} [x_i\ne x_j]$
where $[\cdot]$ is 1 if its argument is true, and 0 otherwise.
Distribution~\eqref{eq:Gibbs} for such a Hamiltonian is called the {\em Ising model}.
It is  {\em ferromagnetic} if $\beta>0$, and {\em antiferromagnetic} if $\beta<0$
(although in the latter case the function $H'(x)=-H(x)$ is usually treated as the Hamiltonian).
Computing $Z(\beta)$ exactly is a $\#$P-complete problem, and is even hard to approximate in the antiferromagnetic
case~\cite{JerrumSinclair:Ising}.
\end{example}
The problem of counting various combinatorial objects such as  proper $k$-coloring and matchings in graphs can also be naturally phrased
as estimating the partition function.
\begin{example}\label{ex:colorings}
Let $\Omega=\{1,\ldots,k\}^{|V|}$ be the set of all colorings in 
an undirected graph $G=(V,E)$.
Define  $H(x)=\sum_{\{i,j\}\in E}[x_i=x_j]$, then $Z(+\infty)$ gives the number of proper $k$-colorings.
\end{example}
\begin{example}\label{eq:matchings}
Let $\Omega$ be the set of matchings $M\subseteq E$ in 
an undirected graph $G=(V,E)$.
Define  $H(M)=|M|$, then $Z(0)=|\Omega|$.
\end{example}

A related problem is that of {\em sampling} from the distribution $\mu_\beta$ for a given
value of $\beta$. There is a vast literature on designing sampling algorithms from Gibbs distributions,
see e.g.~\cite{Metropolis,SwendsenWang,Huber:AAP04,FillHuber:Vervaat} or~\cite{MCMC:handbook} for an overview.
For the ferromagnetic Ising model there exists a polynomial-time approximate sampling algorithm~\cite{JerrumSinclair:Ising}
and also an exact sampling algorithm that appears to run efficiently at or above the {\em critical temperature}~\cite{ProppWilson:96}.
Approximate sampling of $k$-colorings in low-degree graphs is addressed in~\cite{Jerrum:colorings,Vigoda:colorings}
(for $\beta=+\infty$, though techniques are potentially extendable to other values of $\beta$),
and for matchings polynomial-time approximate sampling is described in~\cite[Section 2.3.5]{Matthews:PhD}.

It is known that the ability to sample can be used for designing
a {\em randomized
approximation scheme} for estimating the partition function. By definition, it is an algorithm that for a given $\varepsilon>0$ produces an estimate $\hat Q$ of the desired quantity $Q$ such
that $\hat Q\in\left[\frac{Q}{1+\varepsilon},Q(1+\varepsilon)\right]$ with probability at least $3/4$.
(The value $3/4$ is arbitrary: by repeating the algorithm multiple times and
taking the median of the outputs the probability can be boosted to any other constant in $(0,1)$).
This paper studies the following question: how many samples are needed to approximate $Z(\beta)$ with a given accuracy~$\varepsilon$?

\myparagraph{Formal description}
To state the complexity of different approaches, we need to introduce several quantities.
First, we assume that $H(x)\in\{0\}\cup[1,n]$ for any $x\in\Omega$ where $n$ is a known number.
Non-negativity of the Hamiltonian implies that $Z(\cdot)$ is a decreasing function.
Our goal will be to estimate the ratio  $Q=Z(\bmin)/Z(\bmax)$ for given values $\bmin<\bmax$.
Note that computing $Z(\beta)$ for some specific value of $\beta$ is usually an easy task, so this will allow estimating $Z(\beta)$ for any other $\beta$.
In particular,  in Examples~\ref{ex:Ising},~\ref{ex:colorings} and~\ref{eq:matchings}
we have 
 $Z(0)=2^{|V|}$, $Z(0)=k^{|V|}$ and $Z(+\infty)=1$
respectively.

Let us denote $q=\log Q$, and  assume that there exists an oracle that can produce a sample $X\sim\mu_\beta$ for a given value $\beta\in[\bmin,\bmax]$.
When stating asymptotic complexities, we will always assume that $q=\Omega(1)$, $n=1+\Omega(1)$ and $\varepsilon=O(1)$ to simplify the expressions.
Bez\'akov\'a et al.~\cite{Bezakova08} showed that $Q$ can be estimated
using $O(q^2 (\ln n)^2)$ oracle calls in the worst case (for a fixed $\varepsilon$). 
This was improved to $O(q (\ln q + \ln n)^5\varepsilon^{-2})$ expected number of calls
 by \v{S}tefankovi\v{c} et al.~\cite{Stefankovic:JACM09} 
and then to $O(q\ln n\cdot[\ln q + \ln \ln n+\varepsilon^{-2}])$ by Huber~\cite{Huber:Gibbs}.

The first contribution of this paper is to improve the complexity further to $O(q\ln n \cdot  \varepsilon^{-2}    )$
oracle calls (on average). This is achieved by a better analysis of the algorithm in~\cite{Huber:Gibbs}.
The formal statement of our result is given in Section~\ref{sec:main} as Theorems~\ref{th:main} and~\ref{th:main2}.

In many applications we only have an access to approximate sampling oracles.
Using a standard coupling argument, in Section~\ref{sec:approx} we show
that the same complexity can be achieved with approximate oracles assuming that they are within $O(\frac{\varepsilon^2}{q\ln n})$
variation distance from exact oracles.

As our final contribution, we prove a lower bound of $\Omega(q\cdot \varepsilon^{-2})$ oracle calls
under a natural model of computation. The precise statement of the result is given as Theorem~\ref{th:LowerBound} in Section~\ref{sec:LowerBound}.

\begin{remark}
The assumption that $H(\cdot)$ lies in $\{0\}\cup[1,n]$
can be relaxed using a standard trick. Suppose, for example, that $H(x)\in\{h_{\min},h_{\min}+1,\ldots,h_{\max}\}$
where $h_{\min}$ and $h_{\max}$ are known integers. Let $n=h_{\max}-h_{\min}$.
We claim that the problem can be solved using $O(q'\ln n\cdot \varepsilon^{-2})$
oracle calls (on average), where 
either (i)~$q' = q - (\bmax-\bmin)\cdot h_{\min}$,
or (ii)
$q' = -q + (\bmax-\bmin)\cdot h_{\max}$.

Indeed, to achieve the first complexity, define
new Hamiltonian $H'(x)=H(x)-h_{\min}$.
The partition function for the new problem is $Z'(\beta)=e^{\beta h_{\min}}\cdot Z(\beta)$, and so $q'$
is as defined in (i).
(We use ``primes'' to denote all quantities related to the new problem).
We have $H'(x)\in\{0,1,\ldots,n\}$, so the algorithm claimed above can be applied to give an estimate of $q'$
and thus of $q$. Note that distributions $\mu'_\beta$ and $\mu_\beta$ coincide, and so sampling from $\mu_\beta$ allows to sample from $\mu'_\beta$.

To achieve the second complexity, define $H'(x)=-H(x)+h_{\max}$ and also change the bounds: $\beta'_{\min}=-\bmax$ and 
$\beta'_{\max}=-\bmin$. There holds $Z'(\beta)=e^{-\beta h_{\max}}\cdot Z(-\beta)$, and $q'$ is as defined in~(ii).
We again have $H'(x)\in\{0,1,\ldots,n\}$, and distributions $\mu'_\beta$ and $\mu_{-\beta}$ coincide.
We can now use the same argument as before.
\end{remark}

\section{Background and preliminaries}
We will assume for simplicity that $H(\cdot)\ne const$.
Let us denote $z(\beta)=\ln Z(\beta)$. It can be easily checked that
$$
z'(\beta)=\E_{X\sim\mu_\beta} [ -H(X)]
$$
Since $H(\cdot)$ is non-negative and non-constant, we have $z'(\beta)<0$ for any $\beta$ and thus $z(\cdot)$ and $Z(\cdot)$ are
strictly decreasing functions. It is also known~\cite[Proposition 3.1]{WainwrightJordan} that function $z(\cdot)$ is convex for any $H(\cdot)$,
and in fact strictly convex if  $H(\cdot)\ne const$.

Next, we discuss previous approaches for estimating $Z(\bmin)/Z(\bmax)$,
closely following~\cite{Huber:Gibbs}. 

It is well-known that for given values $\beta_1,\beta_2$ an unbiased estimator $W$ of $Z(\beta_2)/Z(\beta_1)$
can be obtained as follows: first sample $X\sim\mu_{\beta_1}$ and then set $W=\exp((\beta_1-\beta_2)H(X))$. Indeed,
$$
\E[W]=\sum_{x\in\Omega}\frac{\exp(-\beta_1 H(x))}{Z(\beta_1)}\cdot\exp((\beta_1-\beta_2)H(x))
=\sum_{x\in\Omega}\frac{\exp(-\beta_2 H(x))}{Z(\beta_1)}
=\frac{Z(\beta_2)}{Z(\beta_1)}
$$
Applying this estimator directly to $(\beta_1,\beta_2)=(\bmin,\bmax)$ or to $(\beta_1,\beta_2)=(\bmax,\bmin)$ is problematic since it usually
has a huge relative variance. A standard approach to reduce the relative variance is  via  the {\em multistage sampling} method
of 
Valleau and Card~\cite{ValleauCard}. First, a sequence $\bmin=\beta_0\le \beta_1\le \ldots\le\beta_\ell=\bmax$ is selected; it
is called a {\em cooling schedule}.
We then have
\begin{equation*}
\frac{Z(\bmin)}{Z(\bmax)}=\frac{Z(\beta_0)}{Z(\beta_1)}\cdot \frac{Z(\beta_1)}{Z(\beta_2)}\cdot\ldots\cdot \frac{Z(\beta_{\ell-1})}{Z(\beta_{\ell})}
\end{equation*}
Throughout the paper we refer to $[\beta_{i},\beta_{i+1}]$ as ``interval $i$'', where $i\in\{0,1,\ldots,\ell-1\}$.
The ratio $Z(\beta_{i})/Z(\beta_{i+1})$ for each such interval can be estimated independently as described above,
and then multiplied to give the final estimate. Fishman calls an estimate of this form
a {\em product estimator}~\cite{Fishman:94}.
Its mean and variance are given by the lemma below.
In this lemma  we use the following notation:
if $X$ is a random variable then  $\Vrel(X)\eqdef\frac{\E[X^2]}{(\E[X])^2}=\frac{{\tt Var}(X)}{(\E[X])^2}+1$
(the relative variance of $X$ plus 1).

\begin{lemma}[{\cite[page 136]{DyerFrieze:91}}] For $P=\prod_i P_i$ where the $P_i$ are independent,
$$
\E[P]=\prod_i \E[P_i], \qquad \Vrel(P) = \prod_i \Vrel(P_i)
$$
\end{lemma}
 
Using a fixed cooling schedule, Bez\'akov\'a et al.~\cite{Bezakova08} obtained an approximation algorithm that
needs $O(q^2 (\ln n)^2)$ samples in the worst case (for a fixed $\varepsilon$). 
\v{S}tefankovi\v{c} et al.~\cite{Stefankovic:JACM09} asymptotically improved this  to $10^8 q (\ln q + \ln n)^5\varepsilon^{-2}$ 
samples on average. They used an {\em adaptive} cooling schedule where the values
$\beta_i$ depend on the outputs of sampling oracles.
A further improvement to $O(q\ln n\cdot[\ln q + \ln \ln n+\varepsilon^{-2}])$ was given by Huber~\cite{Huber:Gibbs}.
One of the key ideas in~\cite{Huber:Gibbs} was to replace the product estimator with
the {\em paired product} estimator, which is described next.

\subsection{Paired product estimator}
One run of this estimator can be described as follows:
\begin{itemize}
\item sample $X_i\sim\mu_{\beta_i}$ for each $i\in[0,\ell]$
\item for each interval $i\in[0,\ell-1]$ compute
$$
W_i = \exp(-\,\mbox{$\frac{\beta_{i+1}-\beta_i}2$}\; H(X_i)),\qquad 
V_i = \exp (\mbox{$\frac{\beta_{i+1}-\beta_i}2$}\;  H(X_{i+1}))
$$
\item compute $W=\prod_i W_i$ and $V=\prod_i V_i$.
\end{itemize}
An easy calculation (see~\cite{Huber:Gibbs}) shows that
$$
\E[W_i]=\frac{Z(\bar\beta_{i,i+1})}  {Z(\beta_{i})},\quad\;\;
\E[V_i]=\frac{Z(\bar\beta_{i,i+1})}{Z(\beta_{i+1})},\quad\;\;
\E[V_i]/\E[W_i] = \frac{Z(\beta_{i})}{Z(\beta_{i+1})},\quad\;\;
\E[V]/\E[W] = \frac{Z(\bmin)}{Z(\bmax)}
$$
where we denoted $\bar\beta_{i,i+1}=\frac{\beta_i+\beta_{i+1}}2$.
Also,
\begin{equation}
\Vrel(W_i)=\Vrel(V_i)=\frac{Z(\beta_i)Z(\beta_{i+1})}{Z(\bar\beta_{i,i+1})^2},\qquad
\Vrel(W)=\Vrel(V)=\prod_i \frac{Z(\beta_i)Z(\beta_{i+1})}{Z(\bar\beta_{i,i+1})^2} \label{eq:SW}
\end{equation}
Although $\E[V]/\E[W] = \frac{Z(\bmin)}{Z(\bmax)}=Q$, using  $V/W$ as the estimator of $Q$
would be a poor choice since it is biased in general.
Instead,~\cite{Huber:Gibbs} uses the following procedure.

\begin{algorithm}[H]
\DontPrintSemicolon
compute $r$  independent samples of $(W,V)$ as described above\;
take their sample averages $\bar W$ and $\bar V$ and output $\hat Q=\bar V/\bar W$ as the estimator of $Q$
\caption{Paired product estimator. {\bf Input}:  schedule $(\beta_0,\ldots,\beta_\ell)$, integer $r\ge 1$. \label{alg:wrapper}}
\end{algorithm}
 
The argument from~\cite{Huber:Gibbs} gives the following result.
\begin{lemma}\label{lemma:PairedProduct}
Suppose that 
\begin{eqnarray}
\Vrel(W) = \Vrel(V)  &\le& \mbox{$   1+ \frac{1}{2}\gamma r\tilde\varepsilon^{\,2}    $}
\label{eq:Scondition}
\end{eqnarray}
where $\tilde\varepsilon=1-(1+\varepsilon)^{-1/2}=\frac{1}{2}\varepsilon + O(\varepsilon^2)$ and $\gamma>0$.
Then $\P(\hat Q/Q\in (\frac{1}{1+\varepsilon},1+\varepsilon))\ge 1-\gamma$.
\end{lemma}
\begin{proof}
We have $\E[\bar W]=\E[W]$ and ${\tt Var}(\bar W)=\frac{1}{r}{\tt Var}(W)$, and so
 $\Vrel(\bar W)=\frac{1}{r}(\Vrel(W)-1)+1$.
By Chebyshev's inequality, 
$\P(|\bar W/\E[\bar W]-1|\ge \tilde\varepsilon ) \le (\Vrel(\bar W)-1) / \tilde\varepsilon^{\,2} = \frac{1}{r} (\Vrel( W)-1) / \tilde\varepsilon^{\,2} \le \gamma/2$.
Similarly, $\P(|\bar V/\E[\bar V]-1|\ge \tilde\varepsilon ) \le \gamma/2$.

Denote $S=\bar W/\E[\bar W]$ and $T=\bar V/\E[\bar V]$. The union bound gives $\P(\max\{|S-1|,|T-1|\}\ge \tilde\varepsilon) \le \gamma$.
Observe that condition $\max\{|S-1|,|T-1|\}<\tilde\varepsilon$ implies $\{S,T\}\subset(1-\tilde\varepsilon,1+\tilde\varepsilon)\subseteq(\frac{1}{(1+\varepsilon)^{1/2}},(1+\varepsilon)^{1/2})$ 
and thus $\frac{\hat Q}Q=\frac{T}{S}\in(\frac{1}{1+\varepsilon},1+\varepsilon)$.
The claim follows.
\end{proof}

Recall that $\Vrel(W)= \Vrel(V)$ is a deterministic function of the schedule $(\beta_0,\ldots,\beta_\ell)$ (see eq.~\eqref{eq:SW}).
We say that the schedule is {\em good} (with respect to fixed constants $r$ and $\gamma$) if the resulting quantity $\Vrel(W)= \Vrel(V)$ satisfies~\eqref{eq:Scondition}.
Huber presented in~\cite{Huber:Gibbs} a randomized algorithm that produces a good
schedule with probability at least $0.95$ (with respect to $r=\Theta(\varepsilon^{-2})$ and $\gamma=0.2$).
By Lemma~\ref{lemma:PairedProduct}, the output $\hat Q$ of the resulting algorithm lies in $(\frac{Q}{1+\varepsilon},Q(1+\varepsilon))$
with probability at least $0.95\cdot(1-\gamma)>0.75$, as desired.

Huber's algorithm for producing schedule $(\beta_0,\ldots,\beta_\ell)$
is reviewed in the next section.
It makes  $O(q\ln n\cdot [\ln q + \ln \ln n])$ calls to the sampling oracle (on average).
Then in Section~\ref{sec:main} we will describe how to reduce the number of oracle calls
to $O(q\ln n)$
while maintaining the desired guarantees.

\subsection{TPA method}
The algorithm of~\cite{Huber:Gibbs} for producing a schedule is based on the {\em TPA method} of Huber and Schott~\cite{TPA,TPA:journal}.
(The abbreviation stands for the ``Tootsie Pop Algorithm''). Let us review the application of the method
to the Gibbs distribution with a non-negative Hamiltonian $H(\cdot)$.

Its key subroutine is procedure ${\tt TPAstep}(\beta)$ that for a given constant $\beta$
produces a random variable in $[\beta,+\infty]$ as follows:
\begin{itemize}
\item \em
sample $X\!\sim\!\mu_\beta$, draw $U\in[0,1]$ uniformly at random, return $\beta-\ln U/H(X)$
(or $+\infty$ if $H(X)\!=\!0$).
\em
\end{itemize}



The motivation for this sampling rule comes from the following  fact (which we  prove here for completeness).
\begin{lemma}\label{lemma:TPAstep} Consider random variable $U=Z({\tt TPAstep}(\beta))$.
If  $H(\cdot)$ is strictly positive (implying that $Z(+\infty)=0$)
then $U$ has the uniform distribution on $[0,Z(\beta)]$.
If $H(x)=0$ for some $x\in\Omega$ (implying that $Z(+\infty)>0))$
then $U$ has the same distribution as the following random variable $U'$:
sample $U'\in[0,Z(\beta)]$ uniformly at random and set $U'\leftarrow\max\{U',Z(+\infty)\}$.
\end{lemma}
\begin{proof}
It suffices to prove 
$\P({\tt TPAstep}(\beta)\ge \alpha)=Z(\alpha)/Z(\beta)$
for any $\alpha\in[\beta,+\infty)$.
We have
 $$[{\tt TPAstep}(\beta)\ge \alpha]=[\ln U/H(X)\le \beta-\alpha]=[\ln U\le (\beta-\alpha)H(X)]=[U\le \exp((\beta-\alpha)H(X))]$$
Therefore,
\begin{eqnarray*}
\P({\tt TPAstep}(\beta)\ge \alpha)
&=&\sum_{x\in\Omega}\P({\tt TPAstep}(\beta)\ge \alpha|X=x)\P(X=x) \\
&=&\sum_{x\in\Omega}\P(U\le \exp((\beta-\alpha)H(x)))\cdot\frac{\exp(-\beta H(x))}{Z(\beta)} \\
&=&\sum_{x\in\Omega}\exp((\beta-\alpha)H(x))\cdot\frac{\exp(-\beta H(x))}{Z(\beta)} \\
&=&\sum_{x\in\Omega}\frac{\exp(-\alpha H(x))}{Z(\beta)} \;\; = \;\; \frac{Z(\alpha)}{Z(\beta)}
\end{eqnarray*}
\end{proof}


Roughly speaking, the TPA method counts how many steps are needed to get from $\bmin$ to $\bmax$.

\begin{algorithm}[H]
\DontPrintSemicolon
set $\beta_0 = \bmin$, let $\calB$  be the empty multiset\;
\For{$i=1$ {\bf to} $+\infty$}{
	sample $\beta_{i}={\tt TPAstep}(\beta_{i-1})$ \;
	if $\beta_{i}\in[\bmin,\bmax]$ then add $\beta_{i}$ to $\calB$, otherwise output $\calB$ and terminate
}
\caption{One run of TPA. {\bf Output:} a multiset $\calB$ of values in the interval  ${[}\bmin,\bmax{]}$. 
 \label{alg:TPA}
}
\end{algorithm}
The output of Algorithm~\ref{alg:TPA} will be denoted as ${\tt TPA}(1)$,
and the union of $k$ independent runs of ${\tt TPA}(1)$ as ${\tt TPA}(k)$.
For a multiset $\calB$ we define multiset $z(\calB)\eqdef\{z(\beta)\:|\:\beta\in\calB\}$ in a natural way.
 (Recall that $z(\cdot)$ is a continuous strictly decreasing function).
It is known~\cite{TPA,TPA:journal} that $z({\tt TPA}(k))$ is a Poisson Point Process (PPP) on $[z(\bmax),z(\bmin)]$ of rate~$k$,
starting from $z(\bmin)$ and going downwards.
In other words, the random variable ${\bz}=z({\tt TPA}(k))$ is generated by the following process.

\begin{algorithm}[H]
\DontPrintSemicolon
set $z_0=z(\bmin)$, let ${\bz}$ be the empty multiset \;
\For{$i=1$ {\bf to} $+\infty$}{
draw $\eta$ from the exponential distribution of rate $k$ (and with the mean $\frac{1}{k}$), set $z_i=z_{i-1}-\eta$\;
if $z_i\in[z(\bmax),z(\bmin)]$ then add $z_i$ to ${\bz}$, otherwise output ${\bz}$ and terminate
}
\caption{Equivalent process for generating $z({\tt TPA}(k))$.\label{alg:PPP}}
\end{algorithm}

One way to use the TPA method is to simply count the number of points in ${\tt TPA}(k)$.
Indeed, $|{\tt TPA}(k)|$ is distributed according to the Poisson distribution with rate $k\cdot(z(\bmin)-z(\bmax))=k\cdot q$,
so $\frac{1}{k}|{\tt TPA}(k)|$ is an unbiased estimator of $q$.
Unfortunately, obtaining a good estimate of $q$ with this approach requires a fairly large number of samples, namely  
$O(q^2)$ for a given accuracy and  the probability
of failure~\cite{TPA,TPA:journal}. A better application of TPA
was proposed in~\cite{Huber:Gibbs},
where the method was used for generating a schedule $(\beta_0,\ldots,\beta_\ell)$ as follows.

\begin{algorithm}[H]
\DontPrintSemicolon
sample $\calB\sim {\tt TPA}(k)$  \;
sort the values in $\calB$ and then keep every $d$th successive value \;
add values $\bmin$ and $\bmax$ and output the resulting sequence $(\beta_0,\ldots,\beta_\ell)=(\bmin,\ldots,\bmax)$
\caption{Generating a schedule $(\beta_0,\ldots,\beta_\ell)$. {\bf Input:} integers $k,d\ge 1$.\label{alg:sort}}
\end{algorithm}

Note that the resulting sequence $(z_1,\ldots,z_{\ell-1})$=$(z(\beta_1),\ldots,z(\beta_{\ell-1}))$
can be described by a process in Algorithm~\ref{alg:PPP}
where $\eta$ is drawn as the sum of $d$ exponential distributions each of rate $k$;
this is the gamma (Erlang) distribution with shape parameter $d$ and rate parameter $k$.

Huber showed in~\cite{Huber:Gibbs} that if $d=\Theta(\ln q+\ln\ln n)$ and $k=\Theta(d\ln n)$ (with appropriate constants)
then Algorithm~\ref{alg:sort} produces a good schedule with high probability.
Since $q$ is unknown in practice, \cite{Huber:Gibbs}~uses a two-stage procedure: first an estimate $\hat q=2\cdot{\frac{ |{\tt TPA}(5)|}5}+1$
is computed, which is shown to be an upper bound on $q$ with probability at least $0.99$. This estimate is then used for setting $d$ and~$k$.

In the next section we prove that the algorithm has desired guarantees for smaller parameter values, namely $d=\Theta(1)$ and $k=\Theta(\ln n)$.
This allows to reduce the complexity of Algorithm~\ref{alg:sort} by a factor of $\Theta(\ln q+\ln\ln n)$,
and also eliminates the need for a two-stage procedure.

\section{Our results}\label{sec:main}
For technical reasons we will need to make the following assumption for line 2 of Algorithm~\ref{alg:sort}:
if $\beta_1,\beta_2,\ldots$ is the sorted sequence of points in $\calB$ then the index of the first point
to be taken is sampled uniformly from $\{1,\ldots,d\}$ (and after that the index is always incremented by $d$).

Denote $m=\frac{k}d$ and $z_i=z(\beta_i)$ for $i\in[0,\ell]$.
We treat $m$ and $d$ as being fixed, and $k=md$ as their function.
Also let  $\delta=\ln\Vrel(W)=\ln\Vrel(V)$. From~\eqref{eq:SW} we get
\begin{eqnarray}
\delta=\sum_i \delta_i\;, \mbox{~~~~~~~~~~~} \delta_i=z(\beta_{i}) - 2\mbox{$z\left(\frac{\beta_i+\beta_{i+1}}{2}\right)$} + z(\beta_{i+1}).
\end{eqnarray}
Since $z(\cdot)$ is convex, we have $\delta_i\ge 0$ for all $i$.

\myparagraph{\fbox{Case I: $H(x)\in[1,n]$ for all $x\in\Omega$}} First, let us assume that $H(\cdot)$ does not take value $0$.
In this case the proofs become somewhat simpler, and we will get slightly smaller constants.

Huber showed that for  $d=\Theta(\ln(q\ln n))$  
the schedule is {\em well-balanced} with probability $\Theta(1)$,
meaning that {\bf all} intervals $i$ satisfy $z_i-z_{i+1}\le \tau\cdot\frac{1}{m}$ for a constant $\tau=\frac{4}{3}$.
(Note that $\E[z_i-z_{i+1}]\approx\frac{1}m$, ignoring boundary effects).
It was then
%
%
proved\footnote{More precisely, this is what the argument of~\cite{Huber:Gibbs} would give assuming that $H(\cdot)$ does not take value $0$.} 
that a well-balanced schedule satisfies $\delta\le \frac\tau 2 \cdot \frac{\ln n}{m}$,
leading to condition~\eqref{eq:Scondition}.
We improve on this result as follows.

Choose a constant $\tau>0$ (to be specified later), and say that interval $i$ is {\em large} if $z_i-z_{i+1}>\tau\cdot\frac{1}{m}$, and {\em small} otherwise.
Let $\delta^+$ be the sum of $\delta_i$ over large intervals and $\delta^-$ be the sum of $\delta_i$ over small intervals
(so that $\delta=\delta^++\delta^-$).
In Section~\ref{sec:lemma:delta-bound:proof} we prove the following fact. 
(Recall that $\delta^+,\delta^-$ are deterministic functions of the schedule $(\beta_0,\ldots,\beta_\ell)$).
\begin{lemma}\label{lemma:delta-bound}
There holds $\delta^-\le\frac{\tau}2\cdot\frac{\ln n}m$ and 
$
\E[\delta^+]\le  \frac{\Gamma(d+2,\tau d)}{2d\;\cdot \;d!}\cdot\frac{\ln n}m
$
for the schedule $(\beta_0,\ldots,\beta_\ell)$ produced by Algorithm~\ref{alg:sort} with parameters $k=md$ and $d$,
where $\Gamma(\cdot,\cdot)$ is the upper incomplete gamma function:
$$
\Gamma(a,b)=\int_b^{+\infty}t^{a-1}e^{-t}dt\qquad (\mbox{with }\Gamma(a,0)=\Gamma(a)=(a-1)!)
$$
\end{lemma}
Using Markov's inequality, we can now conclude that for any $\tau^+>0$ we have
$$
\P(\delta^+ \ge \mbox{$\frac{\tau^+}2$}\cdot\mbox{$\frac{\ln n}m$})
\;\;\le\;\;\frac{1}{\frac{\tau^+}2\cdot\mbox{$\frac{\ln n}m$}}\cdot \E[\delta^+]
\;\;\le\;\;  \frac{\Gamma(d+2,\tau d)}{\tau^+\;\cdot\; d\;\cdot \;d!}
$$
Thus, with probability at least $1- \frac{\Gamma(d+2,\tau d)}{\tau^+\;\cdot\; d\;\cdot \;d!}$ Algorithm~\ref{alg:sort}
produces a schedule satisfying $\delta \le \frac{\tau+\tau^+}2\cdot\frac{\ln n}m$.

Recall that we want Algorithm~\ref{alg:sort} to succeed with probability at least $\rho=\frac{0.75}{1-\gamma}$
to make the overall probability of success at least $0.75$. (Here $\gamma$ is the constant from Lemma~\ref{lemma:PairedProduct}).
Let us define function $\tau_\rho(d)$ as follows:
$$
\tau_\rho(d) 
= \min_{\tau\ge 0,\tau^+>0} \left\{\tau+\tau^+\:|\: \mbox{$\frac{\Gamma(d+2,\tau d)}{\tau^+\;\cdot\; d\;\cdot \;d!}$} \le 1-\rho\right\}
= \min_{\tau\ge 0} \left[\tau+ \mbox{$\frac{\Gamma(d+2,\tau d)}{(1-\rho)\;\cdot\; d\;\cdot \;d!}$} \right]
$$

\pagebreak

The table below shows some values of this function for $\gamma=0.24$ and $\rho=\frac{0.75}{1-\gamma}=\frac{75}{76}$ (computed with the code of~\cite{GammaIntegral}).

\begin{table}[!h]
  \centering
  \label{tab:table1}
  \begin{tabular}{c||c|c|c|c|c|c|c|c|c|c}
    $d$                           & 1 & 2 & 4 & 8 & 16 & 32 & 64 & 128 & 256 & 512 \\
    \hline
    upper bound on $\tau_\rho(d)$ & 9.903 & 6.052 & 4.000 & 2.860 & 2.197 & 1.794 & 1.539 & 1.372 & 1.260 & 1.184 \\
    achieved with $\tau=$         & 8.645 & 5.384 & 3.634 & 2.653 & 2.075 & 1.720 & 1.492 & 1.342 & 1.241 & 1.170
  \end{tabular}
\end{table}


We can now formulate our main result for case I.
\begin{theorem} \label{th:main}
Let $\hat Q$ be the estimate given by Algorithm~\ref{alg:wrapper} (with parameter $r$)
applied to the schedule produced by Algorithm~\ref{alg:sort} (with parameters $k=md$ and $d$).
Suppose that
\begin{equation}
\hspace{60pt}m \ge \frac{\tau_\rho(d)\cdot \ln n}{2 \ln\left(\mbox{$1+\frac{1}{2}$}\gamma  r\tilde\varepsilon^{\,2}\right)}
\qquad\quad\mbox{for some }\gamma\in(0,0.25)\mbox{ and $\rho=\frac{0.75}{1-\gamma}$}
\label{eq:m:condition}
\end{equation}
where $\tilde\varepsilon=1-(1+\varepsilon)^{-1/2}=\frac{1}{2}\varepsilon + O(\varepsilon^2)$. 
Then $\hat Q\in(\frac{Q}{1+\varepsilon},Q(1+\varepsilon))$ with probability at least $0.75$.
The expected number of oracle calls that this algorithm makes is $mq(r+d)+2r+1$. 

In particular, \eqref{eq:m:condition} will be satisfied for
$d=64$,
$m\ge 3.6 \cdot \ln n$
and
$
r=\left\lceil 2\tilde\varepsilon^{\,-2}\right\rceil=8(1+o(1))\varepsilon^{-2}
$.
\end{theorem}

\begin{proof}
As we just showed,
\begin{equation}
\mbox{$\P\left(\delta\le \frac{\tau_\rho(d)}2\cdot\frac{\ln n}{m}\right)$}\;\;\ge\;\;\rho\label{eq:GADIGKAHSKGJA}
\end{equation}
Condition $\delta\le \frac{\tau_\rho(d)}2\cdot\frac{\ln n}{m}$ implies condition $\delta\le \ln\left(\mbox{$1+\frac{1}{2}$}\gamma  r\tilde\varepsilon^{\,2}\right)$  (by~\eqref{eq:m:condition}),
which is in turn equivalent to $\Vrel(W) \le 1+\frac{1}{2}\gamma r\tilde\varepsilon^{\,2}$.
Thus, from Lemma~\ref{lemma:PairedProduct} we get
\begin{equation}
\mbox{$\P\left(\hat Q\in(\frac{Q}{1+\varepsilon},Q(1+\varepsilon))\;|\;\delta\le \frac{\tau_\rho(d)}2\cdot\frac{\ln n}{m}\right)\ge 1-\gamma$} \label{eq:GAKGJHASGG}
\end{equation}
Multiplying~\eqref{eq:GADIGKAHSKGJA} and~\eqref{eq:GAKGJHASGG} gives the first claim.

A PPP of rate $k$ on an interval $[z(\bmax),z(\bmin)]$ produces $k[z(\bmin)-z(\bmax)]=mdq$ points on average.
Thus, Algorithm~\ref{alg:sort} makes $mdq+1$ oracle calls on average and produces a sequence $(\beta_0,\ldots,\beta_\ell)$ 
with $\E[\ell]=mq+1$. Algorithm~\ref{alg:wrapper} then makes $(\ell+1) r$ oracle calls, i.e.\ $(mq+2)r$ calls on average.
This gives the second claim. 
\end{proof}

\myparagraph{\fbox{Case II: $H(x)\in \{0\}\cup[1,n]$ for all $x\in\Omega$}} 
We now consider the general case. In Section~\ref{sec:lemma:delta-bound:proof2} we prove the following fact.
\begin{lemma}\label{lemma:delta-bound2}
For any constant $\lambda\in(0,1)$ there exists a decomposition $\delta=\delta^-+\delta^+$ with $\delta^-,\delta^+\ge 0$
such that 
$$
\delta^-\le\ln\frac{1}{1-\lambda}+\frac{\tau}2\cdot\frac{2+\ln \frac n\lambda}m 
\qquad\qquad\quad
\E[\delta^+]\le  \frac{\Gamma(d+2,\tau d)}{2d\;\cdot \;d!}\cdot\frac{2+\ln \frac n\lambda}m
$$
\end{lemma}
As in the first case, we conclude from the Markov's inequality  that with probability at least 
$1- \frac{\Gamma(d+2,\tau d)}{\tau^+\;\cdot\; d\;\cdot \;d!}$ Algorithm~\ref{alg:sort}
produces a schedule satisfying $\delta \le \ln\frac{1}{1-\lambda}+\frac{\tau+\tau^+}2\cdot\frac{2+\ln \frac{n}\lambda}m$.
This leads to
\begin{theorem} \label{th:main2}
The conclusion of Theorem~\ref{th:main} holds if 
\begin{equation}
\hspace{20pt}m \ge \frac{\tau_\rho(d)\cdot (2+\ln \frac{n}\lambda)}{2 \ln\left[\left(\mbox{$1+\frac{1}{2}$}\gamma  r\tilde\varepsilon^{\,2}\right)(1-\lambda)\right]}
\qquad\quad\mbox{for some }\gamma\in(0,0.25)\mbox{, $\rho=\frac{0.75}{1-\gamma}$ and $\lambda\in(0,1)$}
\label{eq:m:condition2}
\end{equation}
For example, \eqref{eq:m:condition2} will be satisfied for
$d=64$,
$m\ge 3.6 \cdot (9+\ln n)$
and
$
r=\left\lceil 2\tilde\varepsilon^{\,-2}\right\rceil=8(1+o(1))\varepsilon^{-2}
$ (where we used $\gamma=0.24$ and $\lambda=e^{-7}$).
\end{theorem}

\subsection{Approximate sampling oracles}\label{sec:approx}
So far we assumed that exact sampling oracles $\mu_\beta$ are used.
For many applications, however, we only have approximate sampling oracles $\tilde\mu_\beta$ 
that are sufficiently close to $\mu_\beta$ in terms of the variation distance $||\cdot||_{TV}$ defined via
$$
||\tilde\mu_\beta-\mu_\beta||_{TV}=\max_{A\subseteq \Omega} |\tilde\mu_\beta(A)-\mu_\beta(A)|=\frac{1}{2} \sum_{x\in\Omega}|\tilde\mu_\beta(x)-\mu_\beta(x)|.
$$
The analysis can be extended to approximate oracles using a standard trick (see e.g.\ \cite[Remark 5.9]{Stefankovic:JACM09}).

\begin{theorem}
Let $\hat Q$ be the output of the algorithm with parameters $d,m,r$ satisfying
the conditions of Theorem~\ref{th:main} or~\ref{th:main2} (depending on whether $H(\cdot)\in[1,n]$ or $H(\cdot)\in\{0\}\cup[n]$),
where exact sampling oracles~$\mu_\beta$ are replaced with
approximate sampling oracles~$\tilde\mu_\beta$ satisfying
 $ ||\mu_\beta-\tilde\mu_\beta||_{TV} \le \frac{\kappa}{mq(r+d)+3r+1}$.
Then $\hat Q\in(\frac{Q}{1+\varepsilon},Q(1+\varepsilon))$ 
with probability at least $0.75-\kappa$. 
\end{theorem}

As mentioned in the introduction, probability $0.75-\kappa$ can be boosted to any other probability in $(0.5,1)$
by repeating the algorithm a constant number of times and taking the median (assuming that $\kappa$ is a constant in $(0,0.25)$).
Alternatively, one can tweak parameters in Theorems~\ref{th:main} and~\ref{th:main2} to get the desired probability directly.

\begin{proof}
It is known that there exists a coupling between $\mu_\beta$ and $\tilde\mu_\beta$ such that 
they produce identical samples with probability at least $1-||\tilde\mu_\beta-\mu_\beta||_{TV}\ge 1-\delta$,
where we denoted $\delta=\frac{\kappa}{mq(r+d)+3r+1}$. 
Let ${\mathbb A}$ and $\tilde{\mathbb A}$ be the algorithms that use respectively exact and approximate samples,
where the $k$-th call to $\mu_\beta$ in ${\mathbb A}$ is coupled with the $k$-th call to $\tilde\mu_{\tilde\beta}$ in $\tilde{\mathbb A}$
when $\beta=\tilde\beta$.
We say that  the $k$-th call is {\em good} if the produced samples are identical.
Note, $\P[\mbox{$k$-th call is good}\:|\:\mbox{all previous calls were good}]\ge 1-\delta$, since the conditioning 
event implies $\beta=\tilde\beta$.
Also, if all calls are good then ${\mathbb A}$ and $\tilde{\mathbb A}$ give identical results.

Let $N$ be the number of points inside $[z(\bmax),z(\bmin)]$ produced by the call ${\tt TPA}(md)$ in Algorithm~\ref{alg:sort}.
Then $N$ follows the Poisson distribution of rate $\lambda=mdq$, i.e.\ 
$\P(N=n)=\frac{\lambda^{n}e^{-\lambda}}{n !}$.
Algorithm~\ref{alg:sort} makes $N+1$ oracle calls, and produces a sequence $(\beta_0,\ldots,\beta_\ell)$ 
with $\ell\le \frac{N}{d}+2$.
Thus, the total number of oracle calls is $N+1+(\ell+1)r\le Nc+3r+1$ where $c=1+\frac{r}{d}$.
Denoting $\mu=\lambda(1-\delta)^c$, we can write
\begin{eqnarray*}
\P[\mbox{all calls are good}]
& \!\ge\! & \sum_{n=0}^{\infty} \P(N=n)\cdot (1\!-\!\delta)^{nc+3r+1}
= \sum_{n=0}^{\infty}  \frac{\lambda^{n}e^{-\lambda}}{n !} \cdot (1\!-\!\delta)^{n c+3r+1}
\\
& \!=\! & \sum_{n=0}^{\infty}  \frac{\mu^{n}e^{-\mu}}{n !} \cdot e^{\mu - \lambda}(1\!-\!\delta)^{3r+1}
=e^{\mu-\lambda}(1\!-\!\delta)^{3r+1}
=e^{-\lambda\left(1-(1-\delta)^c\right)}(1\!-\!\delta)^{3r+1} \\
& \!\ge\! & e^{-\lambda\left(1-(1-c\delta)\right)}(1\!-\!\delta)^{3r+1}
\;\ge\; (1\!-\!\lambda c \delta) (1\!-\!\delta)^{3r+1}
\;\ge\; 1 \!-\! \lambda c \delta \!-\! (3r\!+\!1)\delta
\;\ge\; 1 \!-\! \kappa
\end{eqnarray*}
where we used the facts that $(1-x)^c\ge 1-cx$ and $e^{-x}\ge 1-x$ for $x\ge 0$ and $c\ge 1$.
Using the union bound, we obtain the claim of the theorem.
\end{proof}

\section{Proofs}
\subsection{Proof of Lemma~\ref{lemma:delta-bound}}\label{sec:lemma:delta-bound:proof}

We will assume that the sequence $(\beta_0,\ldots,\beta_\ell)$ is strictly increasing (this holds with probability 1).
Accordingly, the sequence $(z_0,\ldots,z_\ell)$ is strictly decreasing.
The following has been shown in~\cite{Stefankovic:JACM09,Huber:Gibbs}.
\begin{lemma} For any $i\in[0,\ell-1]$ there holds
$\delta_i\le z_i-z_{i+1}$ and also
$$
\frac{-z'(\beta_{i})}{-z'(\beta_{i+1})} \ge \exp(2\delta_i/(z_i-z_{i+1})) 
$$ 
\end{lemma}
\begin{proof}
Denote $\bar \beta=(\beta_i+\beta_{i+1})/2$ and $\bar z=z(\bar \beta)$, then $\delta_i=z_{i+1}-2\bar z+z_i$.
 Since $z(\cdot)$ is a convex strictly decreasing function, 
we have
$$
-z'(\beta_{i})\ge\frac{z_i-\bar z}{\bar\beta-\beta_i}\qquad 
-z'(\beta_{i+1})\le\frac{\bar z-z_{i+1}}{\beta_{i+1}-\bar\beta}
$$
Since $\bar\beta-\beta_i=\beta_{i+1}-\bar\beta$, taking the ratio  gives the second claim of the lemma:
$$
\frac{-z'(\beta_{i})}{-z'(\beta_{i+1})}
\ge\frac{z_i-\bar z}{\bar z-z_{i+1}}
=\frac{\frac{1}{2}(z_i-z_{i+1}+\delta_i)}{\frac{1}{2}(z_i-z_{i+1}-\delta_{i})}=\frac{1+\lambda}{1-\lambda}\ge e^{2\lambda}
$$
where we denoted $\lambda=\frac{\delta_i}{z_i-z_{i+1}}\ge 0$ and observed that $\lambda<1$ since $1-\lambda=2\frac{\bar z-z_{i+1}}{z_i-z_{i+1}}>0$.
The fact that $\lambda<1$ also gives the first claim of the lemma.
\end{proof}

Let us define $s(\beta)=\ln [-z'(\beta)]$ and $s_i=\ln [-z'(\beta_i)]$ for $i\in[0,\ell]$,
then function $s(\cdot)$ and the sequence $(s_0,\ldots,s_\ell)$ are strictly decreasing.
Since $z(\beta)$ and $s(\beta)$ are continuous strictly decreasing functions of~$\beta$,
we can uniquely express $z$ via $s$ and define a continuous strictly increasing function $z(s)$ on the interval $S\eqdef[s_\ell,s_0]$
 (Fig.~\ref{fig:plots}(b)). Note, with some abuse of notation we use $z(\cdot)$
for two different functions: one of argument $\beta$, and another one of argument $s$.
The exact meaning should always be clear from the context.

The inequality in the last lemma for an interval $i\in[0,\ell-1]$ can be rewritten as follows:
\begin{eqnarray}
2\delta_i & \le &  (z_i-z_{i+1}) \cdot (s_{i} - s_{i+1}) 
\end{eqnarray}
Equivalently, we have  $2\delta_i\le Area(\Delta_i)$
where $\Delta_i\subseteq[s_\ell,s_0]\times[z_\ell,z_0]$ is the rectangle with the top right corner
at $(s_i,z_i)$ and the bottom left corner at $(s_{i+1},z_{i+1})$ (Fig.~\ref{fig:plots}(b)).
Let $\Delta^+$ be the union of rectangles $\Delta_i$ corresponding to large intervals $i$ (with $|z_i-z_{i+1}|>\tau\cdot\frac{1}{m}$),
and $\Delta^-$ be the union of $\Delta_i$ corresponding to small intervals $i$.
Then $2\delta^+\le Area(\Delta^+)$ and $2\delta^-\le Area(\Delta^-)$.

\begin{figure}[t]
\begin{center}
		\begin{tabular}{c@{\hspace{60pt}}c@{\hspace{60pt}}c} 
			\includegraphics[scale=0.6]{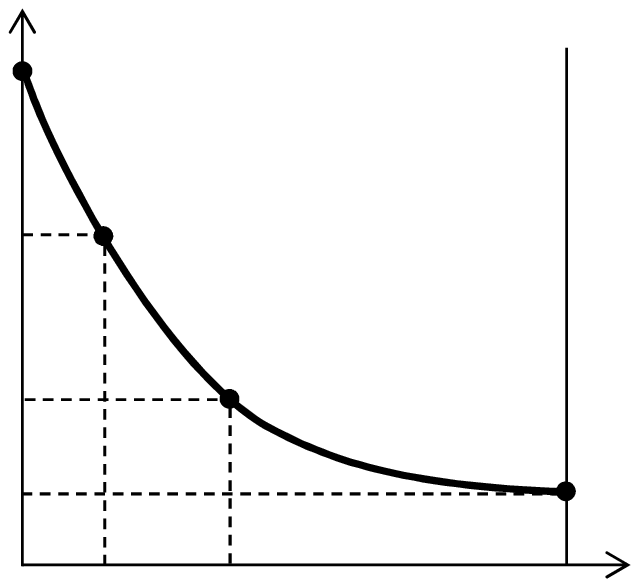} 
			\begin{picture}(0, 0)
				\put(-22,-6){\footnotesize $\bmax$} 
				\put(-115,-6){\footnotesize $\bmin$} 
				\put(-121,87){\footnotesize $z_0$} 
				\put(-121,59){\footnotesize $z_1$} 
				\put(-121,31){\footnotesize $z_2$} 
				\put(-121,14){\footnotesize $z_3$} 
				\put(-8,8){\footnotesize $\beta$} 
				\put(-107,100){\footnotesize $z(\beta)$} 
			\end{picture}
&
			\includegraphics[scale=0.6]{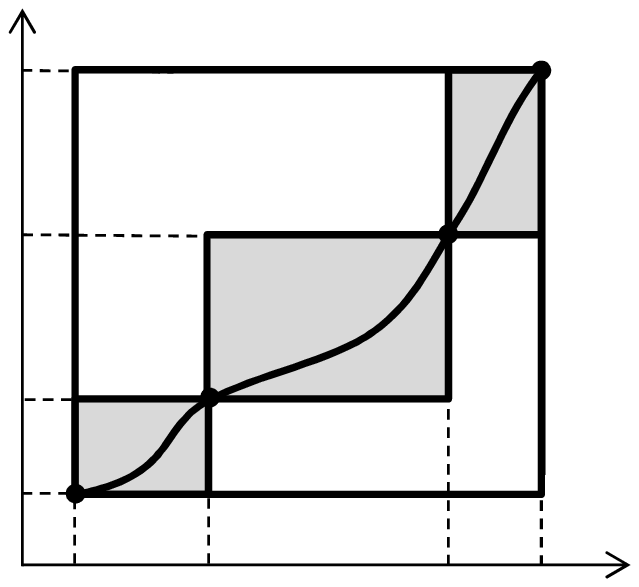} 
			\begin{picture}(0, 0)
				\put(-22,-4){\footnotesize $s_0$} 
				\put(-39,-4){\footnotesize $s_1$} 
				\put(-80,-4){\footnotesize $s_2$} 
				\put(-103,-4){\footnotesize $s_3$} 
				\put(-121,87){\footnotesize $z_0$} 
				\put(-121,59){\footnotesize $z_1$} 
				\put(-121,31){\footnotesize $z_2$} 
				\put(-121,14){\footnotesize $z_3$} 
				\put(-7,7){\footnotesize $s(\beta)$} 
				\put(-107,100){\footnotesize $z(\beta)$} 
				\put(-96,23){\small $-$} 
				\put(-64,47){\small $+$} 
				\put(-35,80){\small $+$} 
			\end{picture}  
&
			\includegraphics[scale=0.6]{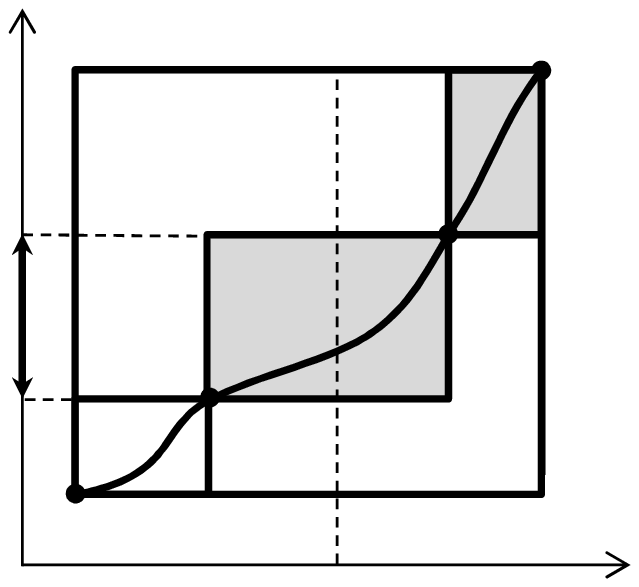} 
			\begin{picture}(0, 0)
				\put(-58,-4){\footnotesize $s$} 
				\put(-125,43){\footnotesize $\eta^+_s$} 
			\end{picture}  \\
~ \vspace{0pt} \\
\small (a) & \small (b) & \small (c) \vspace{-14pt}
		\end{tabular}
	\end{center}
	\caption{(a) $z(\beta)=\ln Z(\beta)$ is a strictly convex decreasing function. Four dots show a possible output of Algorithm~\ref{alg:sort}.
Here $\ell=3$ and $[\beta_0,\beta_\ell]=[\bmin,\bmax]$.
                 (b) Definitions of the sets $\Delta_i$ (in gray). Intervals $0$ and $1$ are assumed to be large, while interval $2$ is small. (c) Definition of the variable $\eta^+_s$ (in the case of a large interval). 
        }
\label{fig:plots}
\end{figure}

By geometric considerations it should be clear that 
$$
Area(\Delta^-)
\;\;\le\;\;  \max\left\{|z_i-z_{i+1}|\::\:\mbox{$i$ is small}\right\}\cdot|S|
\;\;\le\;\; \tau\cdot\mbox{$\frac{1}{m}$}\cdot|S|
$$
Observe that  $-z'(\beta)=\E_{X\sim \mu_\beta}[H(X)]\in[1,n]$ for any $\beta$, and therefore $S=[s_\ell,s_0]\subseteq[0,\ln n]$
and so $|S|\le\ln n$.
This establishes the first claim of Lemma~\ref{lemma:delta-bound}. Next, we focus on proving the second claim.

For a point $s\in S$ let $\eta_s$ be the length of the interval $(z_{i+1},z_i)$
into which $z(s)$ falls (or $0$, if $z(s)\in\{z_\ell,\ldots,z_0\}$).
Also let  $\eta^+_s=\psi[\eta_s]$ where $\psi[\cdot]$ is the following function:
 $\psi[a]=a$ if $a>\tau\cdot \frac{1}m$, and $\psi[a]=0$ otherwise.
Thus, if $z(s)\in(z_{i+1},z_i)$ for some large interval $i$ then $\eta^+_s=z_{i}-z_{i+1}$ 
 (Fig.~\ref{fig:plots}(c)),
otherwise $\eta^+_s=0$.
We have
$$
Area(\Delta^+)
=\int_{S}\eta^+_sds
$$
The linearity of expectation  gives
\begin{eqnarray}
2\E[\delta^+] 
&\le& \E(Area(\Delta^+))
\;\;=\;\; \int_{S} \E[\eta^+_s]ds 
\;\;\le\;\; \max_{s\in S}\E[\eta^+_s]\cdot |S|
\label{eq:expectation-linearity}
\end{eqnarray}

Now let $X_0,X_1,X_2,\ldots$ be a Poisson process on $[0,+\infty)$
and $X_{-1},X_{-2},\ldots$ be a Poisson process on $(-\infty,0]$ (both with rate $k$).
Thus, 
$X_i=\xi_0+\ldots+\xi_{i}$ for $i\ge 0$
and 
 $X_i=-\xi_{-1}-\xi_{-2}-\ldots-\xi_{i}$ for $i\le -1$, where $\xi_j$ are i.i.d.\ variables from the exponential distribution of rate $k$.
By the superposition theorem for Poisson processes~\cite[page 16]{Kingman:PPP},
bidirectional sequence  ${\bf X}=\ldots,X_{-2},X_{-1},X_0,X_1,X_2,\ldots$ 
is a Poisson process  on $(-\infty,+\infty)$ (again with rate $k$),
and in particular it is translation-invariant.

Let ${\bf Y}=\ldots,Y_{-2},Y_{-1},Y_0,Y_1,Y_2,\ldots$ be the following process:
draw an integer $c\in\mbox{$\{0,\ldots,d-1\}$}$ uniformly at random and
then set $Y_i=X_{di+c}$ for each $i$.
It can be seen that  ${\bf Y}$ models 
the output $(\beta_0,\ldots,\beta_\ell)$ of Algorithm~\ref{alg:sort} as follows: 
take the sequence $z(\bmin)-Y_0,z(\bmin)-Y_1,z(\bmin)-Y_2,\ldots$,
restrict to $[z(\bmax),z(\bmin)]$ and  append $z(\bmin)$ and $z(\bmax)$.
Then the resulting sequence has the same distribution as $(z(\beta_0),\ldots,z(\beta_\ell))$.
We assume below that $(\beta_0,\ldots,\beta_\ell)$
is generated by this procedure.

For a point $a\in\mathbb R$ let $\theta_a$ be the length of the interval $(Y_i,Y_{i+1})$ into which $a$ falls
(or $0$, if no such interval exists). Note,
the distribution of random variable $\theta_a$ does not depend on $a$ (since  process ${\bf Y}$
is translation-invariant). We also denote 
 $\theta^+_a=\psi[\theta_a]$, and let $\theta$ and $\theta^+=\psi[\theta]$ be random variables with the same
distributions as  $\theta_a$ and $\theta^+_a$, respectively (for any fixed $a$).
Clearly, for each $s\in[s_\ell,s_0]$ we have $\eta_s\le \theta_a$ and $\eta^+_s\le \theta^+_a$ for a suitably chosen $a$,
namely, $a=z(\bmin)-z(s)$. (Note, if $z(s)\in(z_{\ell-1},z_1)$ then $\eta_s= \theta_a$ and $\eta^+_s= \theta^+_a$,
but at the boundaries the inequalities may be strict).
We thus have
\begin{equation}
\E[\eta^+_s]\le \E[\theta^+]
\label{eq:eta-nu}
\end{equation}

\begin{lemma}\label{lemma:nu-distribution}
Variable $\theta$ has the gamma (Erlang) distribution with shape parameter $d+1$ and rate $k$,
whose probability density is  $f(t)=k^{d+1}t^de^{-kt} / d!$ for $t\ge 0$.
\end{lemma}
\begin{proof}
We prove this fact for variable $\theta_a$ with $a=0$.
We know that $Y_{-1}=X_{c-d}\le 0$ and $Y_0=X_c\ge 0$, so $\theta_0=Y_0-Y_{-1}$ (with probability 1).
By construction, $X_c-X_{c-d}=\xi_{c-d}+\xi_{c-d+1}+\ldots+\xi_{c}$,
i.e.\ $\theta_0$ is a sum of $d+1$ i.i.d. exponential random variables each of rate $k$.
This implies the claim.
\end{proof}
\begin{remark}
It may seem counterintuitive that all intervals $\zeta_i=Y_{i}-Y_{i-1}$ are distributed
as a sum of $d$ exponential random variables with the exception of $i=0$,
in which case it is a sum of $d+1$ variables
(even though ${\bf Y}$ is translation-invariant). This can be viewed as an instance of the ``inspection paradox'',
discussed e.g.\ at~\cite{link:PoissonQuestion}.
Below we describe an alternative approach,
which may help to understand this phenomenon.

Let $\zeta$ be a sum of $d$ exponential random variables each of rate $k$
and $g(\cdot)$ be the probability density of $\zeta$.
Then the following (non-rigorous) argument shows that
the probability density of $\theta$ is 
 $tg(t)/\E[\zeta]$ (after which a simple calculation would prove the claim).

Let $L$ be some large number.
Since the distribution of $\theta_a$ does not depend on $a$,
we can define  $\theta$ as the output of the following process:
sample ${\bf Y}$, sample $a\in[0,L]$ uniformly at random, and then set $\theta=\theta_a$
(i.e.\ the length of the interval in ${\bf Y}$ into which $a$ falls).
Let us compute the probability that $\theta\in[t,t+dt]$. 
Process ${\bf Y}$ will have $L/\E[\zeta]$ intervals in $[0,L]$ on average,
and out of those $(L/\E[\zeta])\cdot (g(t)dt)$ intervals will have length in the range $[t,t+dt]$.
The combined length of such intervals is $(L/\E[\zeta])\cdot (g(t)dt)\cdot t$.
Thus, point $a$ will fall into one of those intervals with probability 
$(L/\E[\zeta])\cdot (g(t)dt)\cdot t / L
=(tg(t)/\E[\zeta])dt$. Therefore, the density of $\theta$ is $tg(t)dt/\E[\zeta]$.
\end{remark}

Recall that $\theta^+=\theta$ if $\theta>\tau/m$, and $\theta^+=0$ otherwise.
Lemma~\ref{lemma:nu-distribution} now gives
\begin{eqnarray*}
\E[\theta^+]
&=& \int\limits_{\tau/m}^{+\infty} tf(t)dt 
\;\;=\;\; \int\limits_{\tau/m}^{+\infty} \frac{k^{d+1}t^{d+1}e^{-kt}}{ d!}dt 
\;\;=\;\; \int\limits_{\tau d/k}^{+\infty} \frac{(kt)^{d+1}e^{-(kt)}}{k\cdot d!}d(kt) \\
&=&\frac{1}{k\cdot d!}\int\limits_{\tau d}^{+\infty} u^{d+1}\, e^{-u} du
\;\;=\;\;\frac{\Gamma(d+2,\tau d)}{k\cdot d!}
\;\;=\;\;\frac{\Gamma(d+2,\tau d)}{md\cdot d!}
\end{eqnarray*}

Combining this with~\eqref{eq:expectation-linearity} and~\eqref{eq:eta-nu} and observing again that $|S|\le \ln n$ finally gives the second claim of Lemma~\ref{lemma:delta-bound}.

\subsection{Proof of Lemma~\ref{lemma:delta-bound2}}\label{sec:lemma:delta-bound:proof2}
We will use the same notation as in the previous section.
Since  $H(\cdot)$ can now take value $0$, we have $-z'(\beta)=\E_{X\sim\mu_\beta}[H(X)]\in[0,n]$
and so $[s_\ell,s_0]\subseteq[-\infty,\ln n]$ (instead of $[s_\ell,s_0]\subseteq[0,\ln n]$, as in the previous section).
We will deal with small values of $s(\beta)$ exactly as in~\cite{Huber:Gibbs}.

Recall that $z'(\beta)$ is a strictly increasing function of $\beta$.
Let $\hat \beta$ be the unique value
with $z'(\hat\beta)=-\lambda$. (If it does not exist, then we take $\hat\beta\in\{-\infty,+\infty\}$ using the natural rule).
Denote  $\hat z=z(\hat\beta)$ and $\hat s=\ln[-z'(\hat\beta)]$.
Now introduce the following terminology for an interval $i\in[0,\ell-1]$:
\begin{itemize}
\item interval $i$ is {\em steep} if  $\beta_{i+1}\le\hat\beta$, or equivalently $s_{i+1}\ge\hat s$;
\item interval $i$ is {\em flat} if  $\beta_{i}\ge\hat\beta$, or equivalently $s_{i}\le\hat s$;
\item interval $i$ is {\em crossing} if  $\hat\beta\in(\beta_{i},\beta_{i+1})$, or equivalently $\hat s\in(s_{i+1},s_{i})$.
\end{itemize}
If steep intervals exist then $\hat\beta\ge\bmin$ and $z'(\hat\beta)\le -\lambda$.
(The inequality may be strict if $\hat\beta=+\infty$).
We thus have $[s_{i+1},s_i]\subseteq[\hat s,s_0]\subseteq[\ln \lambda,\ln n]$ for all steep intervals $i$.
The argument from the previous section gives that 
$$
\sum_{\mbox{\small $i$: $i$ is steep and small}}\delta_i\le \frac{\tau}2 \cdot\frac{\ln\frac n\lambda}m  
\qquad\qquad
\E\left[\sum_{\mbox{\small $i$: $i$ is steep and large}}\delta_i\right]\le \frac{\Gamma(d+2,\tau d)}{2d\;\cdot \;d!}\cdot\frac{\ln\frac{n}\lambda}m
$$
(We just need to assume that $\bmax$ was replaced with $\min\{\bmax,\hat\beta\}$, then we would have $S=[s_\ell,s_0]\subseteq[\ln \lambda,\ln n]$
and $|S|\le\ln\frac{n}\lambda$ instead of $|S|\le\ln n$, the rest is the same as in the previous section).

Let us now consider flat intervals. The argument from~\cite{Huber:Gibbs} gives the following fact.
\begin{lemma} 
The sum of $\delta_i$ over flat intervals $i$ is at most $\ln\frac{1}{1-\lambda}$.
\end{lemma}
\begin{proof}
Assume that flat intervals exist, then $\hat\beta\le\bmax$ and $z'(\hat\beta)\ge -\lambda$.
(The inequality may be strict if $\hat\beta=-\infty$).
Denote $\Omega_0=\{x\in\Omega\:|\:H(x)=0\}$ and $\Omega_+=\{x\in\Omega\:|\:H(x)\ge 1\}$, then $\Omega=\Omega_0\cup\Omega_+$ and
$$
\E_{X\sim\mu_{\hat\beta}}[H(X)]
=\frac{\sum_{x\in\Omega_+}\;\;H(x)e^{-\hat\beta H(x)} }{Z(\hat\beta)}
\ge\frac{\sum_{x\in\Omega_+}\;\;e^{-\hat\beta H(x)} }{Z(\hat\beta)}
=1-\frac{\sum_{x\in\Omega_0}\;\;e^{-\hat\beta H(x)} }{Z(\hat\beta)}
\ge 1-\frac{Z(\bmax)}{Z(\hat\beta)}
$$
On the other hand, $\E_{X\sim\mu_{\hat\beta}}[H(X)]=-z'(\hat\beta)\le \lambda$ and so
$
\frac{Z(\bmax)}{Z(\hat\beta)} 
\ge 1-\lambda
$ and $z(\hat\beta)-z(\bmax)\le \ln\frac 1{1-\lambda}$.
For all flat intervals $i$ we have $[z_{i+1},z_i]\subseteq [z(\bmax),z(\hat\beta)]$
and also $\delta_i\le z_i-z_{i+1}$. This gives the claim of the lemma.
\end{proof}

It remains to consider  crossing intervals.
Let us define values $\delta_c^-$ and $\delta_c^+$ as follows.
If there are no crossing intervals then $\delta_c^-=\delta_c^+=0$.
Otherwise let $i$ be the unique crossing interval;
if $i$ is small then set $(\delta_c^-,\delta_c^+)=(\delta_i,0)$,
and if $i$ is large then set $(\delta_c^-,\delta_c^+)=(0,\delta_i)$.
In all cases we have $\delta_c^-\le \frac{\tau}m$ (since $\delta_i\le z_i-z_{i+1}$).
Also, $\E[\delta_c^+]\le \E[\psi(z_i-z_{i+1})]\le\E[\theta^+]\le \frac{\Gamma(d+2,\tau d)}{md\;\cdot\; d!}$
where function $\psi(\cdot)$ and random variable $\theta^+$ were defined in the previous section.

We can finally prove Lemma~\ref{lemma:delta-bound2}. 
Define $\delta^-$ as $\delta_c^-$ plus the sum of $\delta_i$ over small steep intervals~$i$ and flat intervals $i$.
Define $\delta^+$ as $\delta_c^+$ plus the sum of $\delta_i$ over large steep intervals $i$.
By collecting inequalities above we obtain the desired claim.

\newcommand{\zdiff}{z_{\tt diff}}
\newcommand{\I}{\mathbb I}

\section{Lower bound}\label{sec:LowerBound}
In this section we establish a lower bound on the number of calls to the sampling oracles
for estimating $q=\ln \frac{Z(\bmin)}{Z(\bmax)}$. First, we describe our model of computation and
the set of instances that we allow.

We assume that the estimation algorithm
only receives  values  $H(x)$ from the sampling oracle, and not individual states $x\in\Omega$.
This means that an instance can be defined by counts
$c_h=|\{x\in\Omega\:|\:H(x)=h\}|$ for values $h$ in the range of $H$; these counts uniquely
specify the partition function $Z(\beta)=\sum_h c_h e^{-\beta h}$
and the distribution of sampling oracle outputs for a given $\beta$.
We will thus view an instance as a triplet $\Gamma=(c[\Gamma],\bmin[\Gamma],\bmax[\Gamma])$
where $c[\Gamma]:\mathbb R\rightarrow\mathbb Z_{\ge 0}$ is a function with a finite non-empty support.
When the instance is clear from the context, we will omit the square brackets and write simply $\Gamma=(c,\bmin,\bmax)$.
For a value $h\in{\tt supp}(c)$ let 
$\psi(\beta,h\:|\:\Gamma)$ be the probability that the sampling oracle returns value $h$ when queried at $\beta$
in instance~$\Gamma$:
$$
\psi(\beta,h\:|\:\Gamma)=c_h e^{-\beta h}/Z(\beta)
$$
For a finite subset $\calH\subseteq\mathbb R$ let $\I(\calH)$ be the set
of instances $\Gamma=(c,0,\bmax)$ satisfying ${\tt supp}(c)\subseteq\calH$.
Also for a subset  $\calQ\subseteq\mathbb R$
let $\I(\calH,\calQ)=\{\Gamma\in\I(\calH)\:|\:q^\ast(\Gamma)\in \calQ\}$,
where we denoted $q^\ast(\Gamma)=\ln \frac{Z(0)}{Z(\bmax)}$.

An estimation algorithm $\calA$ applied to instance $\Gamma=(c,0,\bmax)\in\I(\calH)$ is assumed to have the following form.
At step $i$ (for $i=1,2,\ldots$) it does one of the following two actions:
\begin{enumerate}
\item Call the samping oracle for some value $\beta_i\in\mathbb R$.
The oracle then returns a random variable $h_i\in\calH$ with 
 $\P(h_i=h)=\psi(\beta_i,h\:|\:\Gamma)$
for each $h\in\calH$.
\item Output some  estimate $\hat q$ and terminate.
\end{enumerate}
The $i$-th action is a random variable that can depend only on the set ${\tt supp}(c)$, values $\bmin,\bmax$, and on the previously observed sequence $(\beta_1,h_1),\ldots,(\beta_{i-1},h_{i-1})$.
 The output $\hat q$ of the algorithm will be denoted as $q^{\calA}(\Gamma)$,
and the expected number of calls to the sampling oracle as $T^{\calA}(\Gamma)$.

We say that algorithm $\calA$ is an {\em $(\varepsilon,\delta)$-estimator for instance $\Gamma$}
if $\P[|q^{\calA}(\Gamma)-q^\ast(\Gamma)|> \varepsilon]<\delta$. 
We can now formulate our main theorem.
\begin{theorem}\label{th:LowerBound}
There exist positive numbers $q_{\min},n_{\min},c_1,c_2,c_3$ such that the following holds for all 
$q\ge q_{\min}$,
$n\ge n_{\min}$ with $n\in\mathbb Z$, 
$\varepsilon\in(0,c_1 q)$, $\delta\in(0,\frac{1}{4})$.

Denote $m=\left\lceil \frac{c_2\sqrt{q}}{n} \right\rceil$ and $\calH_n^m=\{h\in[1,n]\::\:mh\in\mathbb Z\}$.
Suppose that $\calA$ is an $(\varepsilon,\delta)$-estimator for all instances in $\I(\calH_n^m,\left[\frac{2q}3,\frac{4q}3\right])$.
Then there exists instance $\Gamma\in\I(\calH_n^m,\left[\frac{2q}3,\frac{4q}3\right])$
such that $T^\calA(\Gamma)\ge c_3 q \varepsilon^{-2}\ln \delta^{-1} $.
\end{theorem}

\subsection{Proof of Theorem~\ref{th:LowerBound}}
The proof will be based on the following result. For brevity, we use notation $a\pm b$ to denote the closed interval $[a-b,a+b]$.
\begin{lemma}\label{lemma:LowerBound:1}
Suppose that $\calA$ is an $(\varepsilon,\delta)$-estimator
for instances  $\Gamma\in\I(\calH,\{q\})$ and $\Gamma_1,\ldots,\Gamma_d\in\I(\calH,\mathbb R\setminus(q\pm 2\varepsilon))$, 
where $\bmax[\Gamma_i]=\bmax[\Gamma]$ and ${\tt supp}(c[\Gamma_i])={\tt supp}(c[\Gamma])$ for $i\in[d]$.
Suppose that
\begin{equation}\label{eq:lemma:LowerBound:1}
 \prod_{i\in[d]}\frac{\psi(\beta,h\:|\:\Gamma_i)}{\psi(\beta,h\:|\:\Gamma)} \ge \gamma\qquad\quad\forall \beta\in\mathbb R,h\in{\tt supp}(c[\Gamma])
\end{equation}
for some constant $\gamma\in(0,1)$. Then $T^\calA(\Gamma)\ge \frac{(1-\delta'-\delta)d\ln (\delta'/\delta)}{\ln (1/\gamma)}$
for any constant $\delta'\in[\delta,1-\delta]$.
\end{lemma}

\begin{proof}
The run of the algorithm can be described by a random variable $X=((\beta_1,h_1),\ldots,(\beta_t,h_t),\hat q)$,
where $t$ is the number of oracle calls (possibly infinite, in which case $\hat q$ is undefined).
Let $\calX$ be the set all possible runs, and 
 $\P^\calA(\cdot\:|\:\tilde\Gamma)$ be the probability measure of this random variable conditioned on
$\tilde\Gamma$ being the input instance. The structure of the algorithm implies that this measure can be decomposed as follows:
\begin{equation}
d\P^\calA(x\:|\:\tilde\Gamma)= \psi(x\:|\:\tilde\Gamma)\:d\mu^\calA(x)
\qquad\forall x\in\{((\beta_1,h_1),\ldots,(\beta_t,h_t),\hat q)\in\calX\:|\:\mbox{$t$ is finite}\}
\label{eq:GALSKGHAKSFJASG}
\end{equation}
where $\mu^\calA(\cdot)$ is some measure on $\calX$ that depends only on the algorithm $\calA$, and function $\psi(\cdot)$ is defined via
$$
\psi((\beta_1,h_1),\ldots,(\beta_t,h_t),\hat q\:|\:\tilde\Gamma)=\prod_{i\in[t]} \psi(\beta_i,h_i\:|\:\tilde\Gamma)
$$

Denote $\tau=\frac{d\ln(\delta'/\delta)}{\ln(1/\gamma)}$, and define the following subsets of $\calX$:
\begin{eqnarray*}
\calX^\ast & = & \{((\beta_1,h_1),\ldots,(\beta_t,h_t),\hat q)\in\calX \:\:|\:\: \mbox{$t\le \tau$ and $\hat q\in q\pm\varepsilon$  } \} \\
\calX' & = & \{((\beta_1,h_1),\ldots,(\beta_t,h_t),\hat q)\in\calX \:\:|\:\: \mbox{$t> \tau$ } \} \\
\calX'' & = & \{((\beta_1,h_1),\ldots,(\beta_t,h_t),\hat q)\in\calX \:\:|\:\: \mbox{$t\le \tau$ and $\hat q\notin q\pm\varepsilon$ } \} 
\end{eqnarray*}
Suppose the claim of Lemma~\ref{lemma:LowerBound:1} is false, i.e.\ $T^\calA(\Gamma)\le (1-\delta'-\delta)\cdot\tau$.
We have $T^\calA(\Gamma)\ge \P^\calA(\calX'\:|\:\Gamma)\cdot\tau$, and therefore
\begin{eqnarray}
\P^\calA(\calX'\:|\:\Gamma) & \le & 1-\delta-\delta' \qquad\qquad \label{eq:PcalX:a} 
\end{eqnarray}
Since $\calA$ is a $(\varepsilon,\delta)$-estimator for instances $\Gamma,\Gamma_1,\ldots,\Gamma_d$, we have
\begin{eqnarray}
\P^\calA(\calX''\:|\:\Gamma)~ & < & \delta \label{eq:PcalX:b} \\
\P^\calA(\calX^\ast\:|\:\Gamma_i) & < & \delta \qquad\qquad \forall i\in[d] \label{eq:PcalX:c} 
\end{eqnarray}
Set $\calX$ is a disjoint union of $\calX^\ast,\calX',\calX''$, therefore
$\P^\calA(\calX^\ast\:|\:\Gamma) =1-\P^\calA(\calX'\:|\:\Gamma)-\P^\calA(\calX''\:|\:\Gamma)>1-(1-\delta-\delta')-\delta=\delta'$.
Combining this with~\eqref{eq:PcalX:c} gives
\begin{eqnarray}
\frac{1}d\sum_{i\in[d]} \P^\calA(\calX^\ast\:|\:\Gamma_i) & < & \frac{\delta}{\delta'}\: \P^\calA(\calX^\ast\:|\:\Gamma)\label{eq:PcalX:e} 
\end{eqnarray}

Assumption~\eqref{eq:lemma:LowerBound:1} of the lemma gives that
\begin{equation}
\prod_{i\in[d]}\frac{\psi(x\:|\:\Gamma_i)}{\psi(x\:|\:\Gamma)}\;\;\ge\;\;  \gamma^t \;\;\ge\;\;\gamma^\tau\qquad\quad\forall x=((\beta_1,h_1),\ldots,(\beta_t,h_t),\hat q)\in\calX^\ast
\label{eq:ASGKASJHFASF}
\end{equation}
We can now write
\begin{equation}
\frac{1}{d}\sum_{i\in[d]}\psi(x\:|\:\Gamma_i) 
\;\ge\;\left(\prod_{i\in[d]}\psi(x\:|\:\Gamma_i)\right)^{1/d}
\;\ge\; \gamma^{\tau/d} \psi(x\:|\:\Gamma)
\;=\; \frac{\delta}{\delta'} \:\psi(x\:|\:\Gamma)
\qquad\quad\forall x\in\calX^\ast
\label{eq:GHAKDJFLAKGAS}
\end{equation}
where the first inequality is a relation between arithmetic and geometric means of non-negative numbers,
and the second inequality follows from~\eqref{eq:ASGKASJHFASF}.
We can write
$$
\frac{1}{d}\sum_{i\in[d]}\P^\calA(\calX^\ast\:|\:\Gamma_i)
\;\stackrel{\mbox{\tiny(a)}}=\; \frac{1}{d}\sum_{i\in[d]}\int_{\calX^\ast}\psi(x\:|\:\Gamma_i)\:d\mu^\calA(x)
\;\stackrel{\mbox{\tiny(b)}}\ge\; \frac{\delta}{\delta'} \int_{\calX^\ast}\psi(x\:|\:\Gamma)\:d\mu^\calA(x)
\;\stackrel{\mbox{\tiny(c)}}=\; \frac{\delta}{\delta'}\:\P^\calA(\calX^\ast\:|\:\Gamma)
$$
where (a,c) follow from~\eqref{eq:GALSKGHAKSFJASG} and (b) follows from~\eqref{eq:GHAKDJFLAKGAS}. We obtained a contradiction to~\eqref{eq:PcalX:e}.

\end{proof}

Recall that by definition coefficients of instances should be non-negative integers.
When using Lemma~\ref{lemma:LowerBound:1}, we can relax this requirement
to non-negative rationals (since multiplying coefficients by a constant does not affect quantities in Lemma~\ref{lemma:LowerBound:1}) 
and further to non-negative reals (since they can be approximated by rationals with an arbitrary precision).

We will use Lemma~\ref{lemma:LowerBound:1} with $d=2$ and three instances $\Gamma,\Gamma_+,\Gamma_-$.
First, we will describe the construction of $\Gamma_+$ and $\Gamma_-$ given an instance $\Gamma$.
After stating some properties of this construction, we will define the instance $\Gamma$. 

\myparagraph{Instances $\Gamma_+$ and $\Gamma_-$}
Suppose that $\Gamma=(c,0,\bmax)\in\I(\calH)$.
We set $\Gamma_+=(c^+,0,\bmax)$ and  $\Gamma_-=(c^-,0,\bmax)$ where functions $c^+,c^-$ are given by
\begin{eqnarray*}
c^+_h  =  c_{h}\cdot e^{h\nu},\qquad 
c^-_h  =  c_{h}\cdot e^{-h\nu}\qquad\qquad \forall h\in\mathbb R
\end{eqnarray*}
where $\nu>0$ is some constant.
Let $Z(\cdot)$, $Z_+(\cdot)$, $Z_-(\cdot)$ be the partition functions corresponding to $\Gamma$, $\Gamma_+$, $\Gamma_-$, respectively. One can check that
$$
Z(\beta)=\sum_{h\in{\tt supp}(c)} c_h e^{-\beta h}\qquad Z_+(\beta)=Z(\beta-\nu)\qquad Z_-(\beta)=Z(\beta+\nu)
$$
Denote 
$z(\beta)=\ln Z(\beta)$ and $\zdiff(\beta)=z(\beta)-z(\bmax+\beta)$. Then
$$
q=q^\ast[\Gamma]=\zdiff(0)\qquad
q^\ast[\Gamma_+]=\zdiff(-\nu)\qquad
q^\ast[\Gamma_-]=\zdiff(\nu)
$$
Condition $\Gamma^+,\Gamma^-\in\I(\calH,\mathbb R\setminus(q\pm 2\varepsilon))$ can thus be written as follows:
\begin{eqnarray}
| \zdiff(\pm \nu)- \zdiff(0)|>2\varepsilon
\label{eq:GBAKJSHAKSJGBA}
\end{eqnarray}
Condition~\eqref{eq:lemma:LowerBound:1} after cancellations becomes
$$
\frac{Z^2(\beta)}{Z(\beta-\nu)Z(\beta+\nu)}\ge \gamma\qquad\quad\forall \beta\in\mathbb R
$$
or equivalently
\begin{eqnarray}
z(\beta-\nu)-2z(\beta)+z(\beta+\nu)\le \ln\frac{1}{\gamma} \qquad\quad\forall \beta\in\mathbb R
\label{eq:ALSKGAKSJFAKGA}
\end{eqnarray}
Let us define the following quantities; note that they depend only on instance $\Gamma$:
\begin{equation}
\rho=|\zdiff'(0)|\qquad\qquad
\kappa=\sup_{\beta\in\mathbb R}z''(\beta)
\end{equation}
\begin{lemma}\label{lemma:LowerBound:2}
Let $\Gamma$ be an instance with values $q=q^\ast(\Gamma),\rho,\kappa$ as described above.
Fix $\varepsilon\in(0,\frac{\rho^2}{10\kappa})$.
Suppose that
 algorithm $\calA$ is an $(\varepsilon,\delta)$-estimator for all instances in $\I(\calH,q\pm 4\varepsilon)$.
Then for any constant $\delta'\in[\delta,1-\delta]$ we have
$$
T^\calA(\Gamma)\ge \frac{2(1-\delta'-\delta)\rho^2\ln (\delta'/\delta)}{9\kappa\varepsilon^2}
$$

\end{lemma}
\begin{proof}
Non-negativity of function $c$ implies that function $z(\cdot)$ is convex,
and so $z''(\beta)\in[0,\kappa]$ for all $\beta\in\mathbb R$.
Define $\nu=3\varepsilon/\rho$. For $\beta=\pm \nu$ we can write
$$
|\zdiff(\beta)-\zdiff(0)|
\;\stackrel{\mbox{\tiny(a)}}=\; \left|\zdiff'(0)\beta+\zdiff''(\tilde\beta)\frac{\beta^2}{2}\right|
\;\stackrel{\mbox{\tiny(b)}}\in\; [|\rho \beta| - \kappa \beta^2,|\rho \beta| + \kappa \beta^2]
$$
where in (a) we used Taylor's theorem with the Lagrange form of the remainder (here $\tilde\beta\in\mathbb R$),
and in (b) we used the fact that $|\zdiff''(\tilde\beta)|=|z''(\tilde\beta)-z''(\bmax+\tilde\beta)|\le 2\kappa$.
Observing that $|\rho\beta|=3\varepsilon$ and $\kappa\beta^2=\varepsilon\cdot \frac{9\kappa\varepsilon}{\rho^2}<\varepsilon$,
we get $|\zdiff(\beta)-\zdiff(0)|\in(2\varepsilon,4\varepsilon)$.
Thus, condition~\eqref{eq:GBAKJSHAKSJGBA} holds, and $\Gamma_+,\Gamma_-\in\I(\calH,q\pm 4\varepsilon)$.

Denote $f(\beta)=z(\beta)-z(\beta-\nu)$. Using twice the mean value theorem, we get
$$
z(\beta-\nu)-2z(\beta)+z(\beta+\nu)
= f(\beta+\nu)-f(\beta) 
= f'(\tilde\beta)\nu
=[z'(\tilde\beta)-z'(\tilde\beta-\nu)]\nu
=z''(\tilde{\tilde\beta})\nu^2
\le \kappa\nu^2
$$
where $\tilde\beta,\tilde{\tilde\beta}\in\mathbb R$.
Thus, condition~\eqref{eq:ALSKGAKSJFAKGA} will be satisfied if we set $\gamma\in(0,1)$ so that
$
\ln\frac{1}{\gamma}=\kappa\nu^2=\frac{9\kappa\varepsilon^2}{\rho^2}
$. Lemma~\ref{lemma:LowerBound:2} now follows from  Lemma~\ref{lemma:LowerBound:1}.
\end{proof}

\myparagraph{Instance $\Gamma$}  We now need to construct instance $\Gamma$
such that $q=q^\ast[\Gamma]$ is close to a given value $\bar q$, and the ratio $\frac{\rho^2}{\kappa}$ is large.
We will use an instance with the following partition function:
\begin{eqnarray}
Z(\beta)&=&e^{-\beta}\prod_{k=1}^{N} (a_k+e^{-\beta/m}) 
\end{eqnarray}
where $N$ is some integer in $[m(n-1)]$ and $a_1,\ldots,a_N$ are non-negative numbers. 
Expanding terms yields $Z(\beta)=\sum_{h\in\calH_n^m }c_h e^{-\beta h}$
for some coefficients $c_h\ge 0$, so this is indeed a valid definition of an instance $\Gamma\in\I(\calH_n^m)$. 
In Section~\ref{sec:lemma:LowerBound:3} we prove the following fact.

\begin{lemma}\label{lemma:LowerBound:3}
There exist  values $a_1,\ldots,a_N,\bmax>0$ such that
$q=\frac{\ln 2}2 N^2 \pm O(mN)$ and $\frac{\rho^2}{\kappa}>(\frac{N}4-1)^2$.
\end{lemma}

This will imply Theorem~\ref{th:LowerBound}.
Indeed, let $\bar q$ be the value chosen in Theorem~\ref{th:LowerBound}.
Set $\hat N=\sqrt{\frac{2}{\ln 2}\bar q}$ and
    $N=\left\lceil {\hat N} \right\rceil$. 
Note that
$$
\frac{\hat N}{m(n-1)} \le \frac{\sqrt{\frac{2}{\ln 2}\bar q}}{\frac{c_2\sqrt{\bar q}}{n}(n-1)}
= const\cdot \frac{n}{n-1} \mbox{~~~~~~with~~~~~}const  =      \mbox{$\sqrt{   \frac{2}{\ln 2}  }$} \; / \; c_2
$$
Thus, setting $c_2 > \sqrt{   \frac{2}{\ln 2}  }$ will ensure that $N\in[m(n-1)]$ for sufficiently large $n$.

We have $q=\frac{\ln 2}2 N^2 \pm O(mN)=\frac{\ln 2}2 \hat N^2 \pm O(m\hat N)=\bar q\pm O(m \sqrt{\bar q})=\bar q\left(1\pm O\left(\frac{m}{\sqrt{\bar q}}\right)\right)$.
Recalling that $m=\left\lceil \frac{c_2\sqrt{q}}{n} \right\rceil$, we
conclude that 
$q\in\left[\frac{3\bar q}{4},\frac{5\bar q}{4}\right]$ if $\bar q,n$ are sufficiently large.
Furthermore, we have $\frac{\rho^2}{\kappa}>(\frac{N}4-1)^2>\frac 16 \bar q$ if $\bar q$ is sufficiently large
(note that $\frac 16<\frac{1}{8\ln 2}$).

We set $c_1=\frac{1}{60}$, so that $\varepsilon\in(0,\frac{1}{60}\bar q)$. Now suppose that the preconditions of Theorem~\ref{th:LowerBound} hold.
It can be checked that $\varepsilon\in (0,\frac{\rho^2}{10\kappa})$ and $q\pm 4\varepsilon\subseteq \left[\frac{2\bar q}3,\frac{4\bar q}3\right]$,
so the preconditions of Lemma~\ref{lemma:LowerBound:2} hold as well.
Setting $\delta'=\frac{1}{2}$ and recalling that $\delta\in(0,\frac{1}{4})$, we obtain the desired result:
$$
T^\calA(\Gamma)\ge \frac{2(1-\frac{1}{2}-\delta)\ln (\frac{1}{2}/\delta)}{9\varepsilon^2}\cdot \frac{\rho^2}{\kappa}
\;\ge\;\frac{\ln \delta^{-1}-\ln 2}{18\varepsilon^2}\cdot \frac{\bar q}6
\;\ge\;\frac{1-\frac{\ln 2}{\ln 4}}{18\cdot 6}\cdot \frac{\bar q\ln \delta^{-1}}{\varepsilon^2}
$$

\subsection{Proof of Lemma~\ref{lemma:LowerBound:3}}\label{sec:lemma:LowerBound:3}

Denote $u=u(\beta)=e^{-\beta/m}$ and $\eta=u(\bmax)=e^{-\bmax/m}$. (The choice of $\eta\in(0,1)$ will be specified later).
We can write
\begin{eqnarray}
z(\beta)=-\beta+\sum_{k=1}^N \ln (a_k+u) \qquad\quad
z'(\beta)=-\sum_{k=N}^N \frac{u}{m(a_k+u)} \qquad\quad
z''(\beta)=\sum_{k=1}^N \frac{a_ku}{m^2(a_k+u)^2}
\end{eqnarray}
\begin{eqnarray}
q&=&z(0)-z(\bmax)\;\;=\;\;m\ln\frac{1}{\eta}+\sum_{k=1}^N \ln \frac{a_k+1}{a_k+\eta} \\
\rho&=&|z'(0)-z'(\bmax)|\;\;=\;\;\frac{1}{m}\sum_{k=1}^N \left[\frac{1}{a_k+1}-\frac{\eta}{a_k+\eta}\right] 
\end{eqnarray}
As for $\kappa=\max_{\beta\in\mathbb R}z''(\beta)$, we will use the following bound.
\renewcommand{\ell}{r}
\begin{lemma}
Suppose that $a_1\ge a_2\ge \ldots \ge a_N> 0$. Then $\kappa\le\max_{\ell\in[N-1]}\kappa_\ell$
where we denoted
$$
\kappa_{\ell}
\;\;=\;\; \frac{1}{m^2}\left[\sum_{k=1}^{\ell} \frac{a_{\ell}}{a_k}
         \;+\; \sum_{k=\ell+1}^N \frac{a_k}{a_{\ell+1}}\right]
$$
\end{lemma}
\begin{proof}
We need to show that 
$$
\sum_{k=1}^N \frac{a_ku}{m^2(a_k+u)^2} 
\;\;\le\;\;   \max_{\ell\in[N-1]}\kappa_{\ell}
\qquad\qquad\forall u\in(0,+\infty)
$$
By taking the derivative one can check that function $g_k(u)=\frac{a_ku}{(a_k+u)^2}$ is increasing on $[0,a_k]$ and decreasing on $[a_k,+\infty)$
(with the maximum at $u=a_k$). Therefore, function $g(u)=\sum_{k=1}^N g_k(u)$ attains a maximum at $[a_N,a_1]$. We can thus assume
w.l.o.g.\ that $u\in[a_N,a_1]$.

Let $\ell\in[N-1]$ be an index such that $u\in[a_{\ell+1},a_\ell]$. For $k\in[1,\ell]$ we have
$
g_k(u)\le g_k(a_\ell)=\frac{a_\ell/a_k}{(1+a_\ell/a_k)^2}\le \frac{a_\ell}{a_k}
$,
and for $k\in[\ell+1,N]$ we have
$
g_k(u)\le g_k(a_{\ell+1})=\frac{a_k/a_{\ell+1}}{(1+a_k/a_{\ell+1})^2}\le \frac{a_k}{a_{\ell+1}}
$. By summing these inequalities we get that $g(u)\le m^2\kappa_\ell$.
\end{proof}

We can now prove Lemma~\ref{lemma:LowerBound:3}. Define
$
a_k=2^{1-k}$ and $\eta=2^{1-N}
$.
For each $k\in[N]$ we have $\ln\frac{a_k+1}{a_k+\eta}\ge \ln\frac{1}{2a_k} = (k-2)\ln 2$ and
 $\ln\frac{a_k+1}{a_k+\eta}<\ln\frac{a_k+1}{a_k} = \ln \frac{1}{a_k} + \ln (1+a_k) \le \ln \frac{1}{a_k} + a_k=(k-1)\ln 2+2^{1-k}$,
therefore
\begin{eqnarray*}
 q&>& m(N-1)\ln 2 + \sum_{k=1}^N (k-2)\ln 2 \;\;=\;\;  \left(m+\frac{N}{2}\right)(N-1)\ln 2 - N\ln 2 \\
 q&<& m(N-1)\ln 2 + \sum_{k=1}^N \left[  (k-1)\ln 2+2^{1-k} \right] \;\;<\;\; \left(m+\frac{N}{2}\right)(N-1)\ln 2 + 2 
\end{eqnarray*}
The following inequalities imply the last two claims of Lemma~\ref{lemma:LowerBound:3}:
\begin{eqnarray*}
\rho&>&\frac{1}{m}\sum_{k=1}^N\left[\frac{1}{1+1}-\frac{\eta}{a_k}\right]
\;\;=\;\;\frac{1}{m}\left[\frac{N}2-\frac{2^N-1}{2^{N-1}}\right]
\;\;>\;\; \frac{1}{m}\left[\frac{N}2-2\right]
\\
\kappa_\ell&<&
\frac{1}{m^2}\left[\sum_{k=1}^\ell\frac{a_\ell}{a_k}+
\sum_{k=\ell+1}^{+\infty} \frac{a_k}{a_{\ell+1}}\right]
\;\;=\;\;\frac{1}{m^2}\left[\frac{2^\ell-1}{2^{\ell-1}}+2\right]
\;\;<\;\; \frac{4}{m^2}\qquad\qquad \forall \ell\in[N-1]
\end{eqnarray*}

\section*{Acknowledgements}
I thank Laszlo Erd\"os and Alexander Zimin   for useful discussions.
In particular, the link~\cite{link:PoissonQuestion}  provided by Alexander
helped with the argument in Section~\ref{sec:lemma:delta-bound:proof}.
The author is supported by the European Research Council under the European Unions Seventh Framework Programme (FP7/2007-2013)/ERC grant agreement no 616160.

\small
\bibliographystyle{plain}
\bibliography{sampling}

\end{document}